%% file: main.tex
\documentclass[10pt, conference, letterpaper]{IEEEtran}
\usepackage{amssymb}
\usepackage{amsmath}
\usepackage{amsthm}
\usepackage{tikz}
\usepackage{subfigure}
\usepackage{epsfig}
\usepackage{graphicx}
\usepackage{comment}

\newcommand{\G}{{\cal G}}
\newcommand{\V}{{\cal V}}
\newcommand{\E}{{\cal E}}
\newcommand{\X}{{\cal X}}
\newcommand{\A}{{\cal A}}
\newcommand{\Z}{{\mathbb Z}}
\newcommand{\var}{ {\rm Var}}
\newcommand{\pb}{p_{ {\rm B}}}
\newcommand{\pbu}{f_{ {\rm B}}}
\newcommand{\zot}{Z^{ {\rm OT}}}
\newcommand{\zunif}{Z^{ {\rm U}}}
\newcommand{\Lunif}{L_{ {\rm U}}}
\newcommand{\Runif}{R_{ {\rm U}}}
\newcommand{\Lot}{L_{ {\rm OT}}}
\newcommand{\Rot}{R_{ {\rm OT}}}

\newtheorem{proposition}{Proposition}
\newtheorem{lemma}{Lemma}

\includecomment{comment}

\begin{document}
\bibliographystyle{IEEEtran}
\title{On Stochastic Estimation of Partition Function}
\author{
 \IEEEauthorblockN{Ali Al-Bashabsheh \hspace{1.5cm} Yongyi Mao}
				\IEEEauthorblockA{School of Electrical Engineering and Computer Science \\
						 University of Ottawa, Canada\\
						 \{aalba059, yymao\}@eecs.uottawa.ca}
}
\maketitle
\pagestyle{plain}
\input{abs.tex}

\input {intro}

\input{prelim}

\input{convergence}


\input{simulations}

\input{conclude}


\bibliography{/home/bat/Dropbox/bibliographys/Holographic_revised} 

\end{document}

%% file: abs.tex
\begin{abstract}
In this paper, we show analytically that the duality of normal factor graphs (NFG) can facilitate stochastic estimation
of partition functions. In particular, our analysis suggests that for the $q-$ary two-dimensional nearest-neighbor Potts
model, sampling from the primal NFG of the model and sampling from its dual exhibit opposite behaviours with respect to
the temperature of the model.  For high-temperature models, sampling from the primal NFG gives rise to better estimators
whereas for low-temperature models, sampling from the dual gives rise to better estimators. This analysis is validated
by experiments. 
\end{abstract}

%% file: intro.tex
\section{Introduction}

The estimation of partition function for statistical models is of fundamental importance in statistical physics, machine
learning and information theory \cite{Wainright:Bounds, Neal:mcm}.  The models we consider in this
paper are specified by a collection of random variables $\{X_i:i=1, 2, \ldots, N\}$, for some positive integer ${N}$;
each random variable $X_i$ is assumed to take values (often called \emph{spins}) from some finite set $\X$; every configuration $x\in \X^{N}$ is
associated with an energy level $E(x)$, and the joint distribution of random variables
$\{X_i:i = 1, 2, \ldots, N\}$ is modelled as the Boltzmann distribution
\begin{equation}
\pb(x):=\frac{e^{-\beta E(x)}}{Z},
\label{eq:model}
\end{equation}
for all $x \in \X^{N}$.
In (\ref{eq:model}), $\beta:=\frac{1}{kT}$ is often referred to as the ``inverse temperature'',
where $T$ is the temperature and $k$ is the Boltzmann constant, and the normalizing constant
$Z:=\sum_{x\in \X^{N}}e^{-\beta E(x)}$ is known as the \emph{partition function}. 

Given $\beta$ and the energy function $E(\cdot)$, exact computation of the partition function $Z$ for systems involving
a large number of random variables is known to be intractable, and it is precisely the intractability of this problem
that roots the hardness of various problems in coding and information theory (e.g., determining the capacity of
constrained codes).  Developing bounding techniques (e.g. \cite{Wainright:Bounds}) and
approximation methods \cite{Potamianos:1997} for estimating the partition functions is thus an active area of
research. 

This work is motivated by the recent empirical observation of \cite{Mehdi:Ising} that for the  
two-dimensional nearest-neighbor Ising model (binary spins), 
the duality of normal factor graphs (NFG) \cite{NFG:hol} appears to facilitate
the estimation of the partition function.
In particular, they experimentally show that for large $\beta$, two stochastic estimation methods
(the Ogata-Tanemura method \cite{Ogata:estimator} based on Gibbs sampling and a method based on uniform sampling)
provide better estimation of the partition function when sampling from the dual NFG compared to sampling from the primal NFG.

In this paper, we explain the behaviour observed in \cite{Mehdi:Ising} and show both analytically and experimentally
that such a trend extends beyond the Ising model to $q$-ary spins, i.e., the standard Potts model \cite{Wu:Potts}.
Along
our development, we also provide insights on the question for what other two-dimensional nearest-neighbor models such
behaviour may hold.

The remainder of this paper is organized as follows.  In Section \ref{sec:prelim}, we precisely state the model
considered in this paper and present the NFG representation of the model and the duality result therein
\cite{NFG:hol}.  A concise review of two stochastic estimation methods  and the approach of \cite{Mehdi:Ising} are
also given in Section \ref{sec:prelim}.  Section \ref{sec:converge} presents an analysis of the convergence behaviour of
the two methods, suggesting that in high $\beta$ regime, sampling from the dual NFG model performs better whereas in
low $\beta$ regime, sampling primal NFG performs better. The analysis is supported by the experimental results presented
in Section \ref{sec:experiments}. The paper is concluded in Section \ref{sec:conclude}, where we extend the results  beyond the Potts model.

%% file: prelim.tex
\section{Preliminaries}

\label{sec:prelim}

\subsection{Model}
\label{subsec:model}

In Equation (\ref{eq:model}), we consider that each index in $\{1, 2, \ldots, N\}$ corresponds to a grid point in an
$L\times L$ square lattice. 
We assume that the lattice is ``wrapped around'' in the sense that the left-most point of each row is connected to the right-most point of
the same row and the top-most point of each column is connected to the bottom-most point of the same column. Let ${\cal
A}$ denote the set of all pairs of adjacent lattice points. The energy function is assumed to take the form
\begin{equation}
\label{eq:energy}
E(x):=-\sum_{\{i, j\} \in \A} g_{ij}(x_i,x_j),
\end{equation}
for a collection of functions $\{g_{ij}: (i,j)\in {\cal A}\}$. 
Such a model is referred to as a two-dimensional nearest-neighbor model. 

We will further assume that the alphabet $\X$ is the abelian group $\Z_{q}:=\{0, \cdots, q-1\}$ and that
\begin{equation}
\label{eq:g}
g_{ij}(x, x')=g(x,x'):=\left\{ \begin{array}{rc} 1, & x = x' \\ -1, & x \neq x'. \end{array} \right.
\end{equation}

Equations (\ref{eq:model}) to (\ref{eq:g})  define a (two-dimensional nearest-neighbor) Potts model.%
\footnote{
We slightly deviate from the traditional definition of the Potts model where the function $g$ is usually
assumed to take the value $0$ instead of $-1$. Without altering the nature of the problem, 
this choice of function $g$ includes the Ising model as the special case of
$q=2$. 
}
(Some authors use the term \emph{standard} Potts model to make explicit the distinction from the ``clock'' model.)
To facilitate later discussions,  we use $\pbu$ to denote $e^{-\beta E(x)}$ in (\ref{eq:model}) and refer to it as the ``unnormalized Boltzmann distribution''.

\subsection{NFG Representation and Duality}

A normal factor graph (NFG) $\G$ is a graph $(\V,\E)$ where each edge $e \in \E$ is
associated a variable $x_e$, and each vertex $v \in \V$ is associated a local function $f_v(x_{E(v)})$, where $E(v)$ is the
set of edges incident with $v$, and for any set ${\cal A}$, $x_{ {\cal A}}:=\{x_a:a\in {\cal A}\}$.
Let $\X_{\G}$ be the support of of the function defined as the multiplication of all local functions, and let
$f_{\G}$ be the restriction of such function to $\X_{\G}$. Further, we define $Z_{\G}$ as the sum of $f_{\G}$ over
$\X_{\G}$, and write $p_{\G}:=f_{\G}/Z_{\G}$. Note that if all the local functions are nonnegative, then $p_{\G}$ is a
probability distribution over $\X_{\G}$. In this case, in alignment with the previous discussions, we refer to
$p_{\G}$, $f_{\G}$, and
$Z_{\G}$ as the distribution, unnormalized distribution, and partition function of the NFG, respectively. We note that
the above definitions of NFG and related terms deviate slightly from those in \cite{NFG:hol}.  This is to simplify our
presentation and exclude the concepts irrelevant to this paper.


It is natural to associate with the model defined in Section \ref{subsec:model} an NFG 
as in Fig.~\ref{fig:pnfg1} (wrapping around is not shown). In the figure, 
each function marked by ``$=$" is an ``equality indicator function'', namely, a function that evaluates to $1$ if all 
it arguments are equal and evaluates to $0$ otherwise;  each equality indicator function corresponds to a random variable in the model.
The function $h$ in the figure is defined by $h(x, x'):=e^{\beta g(x, x')}$. It is not hard to see that the unnormalized distribution, distribution and partition functions associated with this NFG are 
respectively $f_B$, $p_B$ and $Z$ of the model defined by equations (\ref{eq:model}), (\ref{eq:energy}) and (\ref{eq:g}).


Noting that function $g$ only depends on the difference between its arguments, we may express $h$ 
by $h(x, x'):=\kappa(x-x')$, where
\begin{eqnarray}
	\kappa(x) = \left\{ \begin{array}{lc} e^{\beta}, & x = 0 \\ e^{-\beta}, & x \neq 0. \end{array} \right.
	\label{eq:ham}
\end{eqnarray}

Using function $\kappa$, the NFG in Fig.~\ref{fig:pnfg1} may be converted to the NFG in Fig.~\ref{fig:pnfg2} without changing its unnormalized distribution, distribution and partition function.  This latter NFG, which we denote by ${\cal G}$ is in fact preferred in the context of this paper, since the results of this paper depend crucially on a property of $\kappa$, which will become clear momentarily.


\begin{figure}[ht]
\def\s{0.82}	
\def\c{.56}	
\def\L{3}	
\centering
\begin{tikzpicture}[scale=\s, every node/.style={draw, fill=white, minimum size=3mm, inner sep =.2mm}]
    \draw (-.8,-1.8) grid[step=2cm] ({2*(\L+1)-.1},{2*\L+.8});
\small
\node[draw=none, label=below:$h$] at (1,6){}; \node[draw=none, label=below:$h$] at (3,6){}; \node[draw=none, label=right:$h$] at (0,5){};
\scriptsize
    \node[draw=none] at (-.5,6.25){$X_1$}; \node[draw=none] at (1.5,6.25){$X_2$}; \node[draw=none] at (5.5,0.25){$X_{16}$};
    \foreach \x in {0,...,\L}{
      \foreach \y in {0,...,\L}{
        \node at (2*\x,2*\y){$=$};
      }
    }
    \foreach \x in {0,...,\L}{
    \foreach \y in {0,...,\L}{
		\node at (2*\x,2*\y-1){$$};
      }
    }
    \foreach \x in {0,...,\L}{
      \foreach \y in {0,...,\L}{
		\node at (2*\x+1,2*\y){$$};
      }
    }
	\end{tikzpicture}
	\caption{An NFG representing the model specified by  (\ref{eq:model}), 
	(\ref{eq:energy}) and (\ref{eq:g}).}
	\label{fig:pnfg1}
\end{figure}
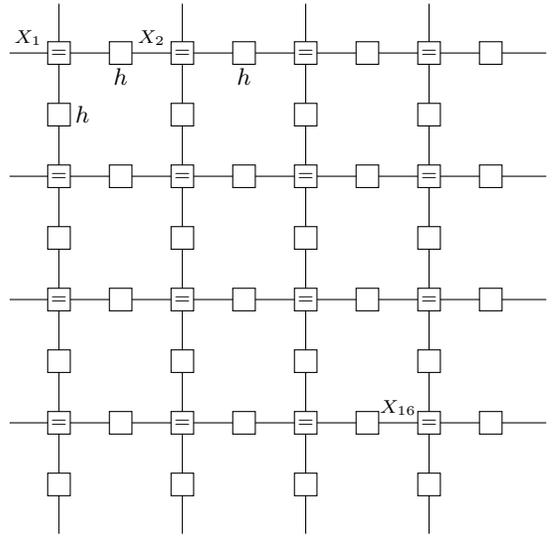

\begin{figure}[ht]
\def\s{0.8}	
\def\c{.56}	
\def\L{3}	
\centering
\begin{tikzpicture}[scale=\s, every node/.style={draw, fill=white, minimum size=3mm, inner sep =.2mm}]
    \small
    \node[draw=none, label=right:$\kappa$] at (1,6.5){}; \node[draw=none, label=right:$\kappa$] at (3,6.5){}; \node[draw=none,
    label=above:$\kappa$] at (.5,5){};
    \draw (-.8,-1.8) grid[step=2cm] ({2*(\L+1)-.1},{2*\L+.8});
    \foreach \x in {0,...,\L}{
      \foreach \y in {0,...,\L}{
        \node at (2*\x,2*\y){$=$};
      }
    }
    \foreach \x in {0,...,\L}{
    \foreach \y in {-0,...,\L}{
		\draw (2*\x,2*\y-1) -- (2*\x+\c,2*\y-1);
		\node[minimum size=1mm] at (2*\x,2*\y-0.5){\tiny $+$};
		\node at (2*\x,2*\y-1){$+$};
		\node at (2*\x+\c,2*\y-1){};
      }
    }
    \foreach \x in {0,...,\L}{
      \foreach \y in {0,...,\L}{
		\draw (2*\x+1,2*\y) -- (2*\x+1,2*\y+\c);
		\node[minimum size=1mm] at (2*\x+0.5,2*\y){\tiny $+$};
		\node at (2*\x+1,2*\y){$+$};
		\node at (2*\x+1,2*\y+\c){};
      }
    }
	\end{tikzpicture}
	\caption{The NFG $\G$.}
	\label{fig:pnfg2}
\end{figure}

It is possible to introduce duality to NFG via the Fourier transform. Briefly, the Fourier transform of any function $f$ on ${\mathbb Z}_q^m$ is another function $\widehat{f}$ defined 
on ${\mathbb Z}_q^m$. In particular, the Fourier transform of an equality indicator function is, up to scale, a ``parity-check'' indicator function, namely a function that evaluates to $1$ if its argument sums to $0$ and evaluates to $0$ otherwise.  A parity-check indicator function is marked by ``$+$'' in an NFG.  Further, the Fourier transform of the function $\kappa$ is 
\begin{eqnarray}
\label{eq:kappa_ft}
	\widehat{\kappa}(x) = \left\{ \begin{array}{lc} e^{\beta}+(q-1)e^{-\beta}, & x = 0 \\ e^{\beta}-e^{-\beta}, & x
		\neq 0. \end{array} \right.
	\label{eq:dham}
\end{eqnarray}

Given an NFG, the dual NFG may be obtained by converting each local function to its Fourier transform and then inserting
a parity-check indicator function to each edge.  It can then be verified that 
the dual NFG of ${\cal G}$ is the NFG ${\cal G}'$ shown in Fig.~\ref{fig:dnfg}.  A duality theorem (generalized Holant
theorem) of NFG \cite{NFG:hol} states, in the context of our model, that
$Z_{\G'}=Z_{\G}/q^{N}$.

\subsection{Estimating Partition Function by Sampling NFG}

Given an NFG $\G$ representing a statistical model, its partition function $Z_{\G}$ may be estimated via evaluating its
unnormalized distribution $f_{\G}$ at a set of configurations $Y_1, Y_2, \ldots, Y_M$ randomly drawn from $\X_{\G}$.
If these configurations are obtained by sampling the distribution $p_{\G}$ (which in practice can be done by Gibbs
sampling), then 
the Ogata-Tanemura (OT) \cite{Ogata:estimator} estimator can be defined as
\begin{eqnarray}
	\zot_{\G}(M):=\frac{|\X_{\G}|}{\frac{1}{M}\sum\limits_{i=1}^{M}\frac{1}{f_{\G}(Y_{i})}},
\label{eq:zot}
\end{eqnarray}
If these samples are drawn uniformly from $\X_{\G}$, an estimator, which we call the ``uniform estimator'', can be defined as
\begin{eqnarray}
	\zunif_{\G}(M):=\frac{|\X_{\G}|}{M}\sum_{i=1}^{M}f_{\G}(Y_{i}),
\label{eq:zunif}
\end{eqnarray}
It can be shown that as $M$ increases, both $\zot_{\G}(M)$ and $\zunif_{\G}(M)$ converges to $Z_{\G}$. 

Given the NFG $\G$  in Fig. \ref{fig:pnfg2} that represents the Potts model, it is easy to see that every local function
in the dual NFG $\G'$ in Fig. \ref{fig:dnfg} is non-negative. The dual NFG $\G'$ may then be regarded also as a statistical
model and the above two estimators may be used to estimate the partition function $Z_{\G'}$, a scaled version of
$Z_{\G}$. This technique was first used in \cite{Mehdi:Ising} for Ising model (Potts model with $q=2$), where the authors
show empirically that at high temperature, both OT estimator and uniform estimator 
give more accurate estimates on the dual NFG.


\begin{figure}[ht]
\def\s{0.8}	
\def\c{.56}	
\def\L{3}	
\centering
\begin{tikzpicture}[scale=\s, every node/.style={draw, fill=white, minimum size=3mm, inner sep =.2mm}]
	\small
	\node[draw=none, label=right:$\widehat{\kappa}$] at (1,6.5){}; \node[draw=none, label=right:$\widehat{\kappa}$] at (3,6.5){}; \node[draw=none,
	label=above:$\widehat{\kappa}$] at (.5,5){};
	\draw (-.8,-1.8) grid[step=2cm] ({2*(\L+1)-.1},{2*\L+.8});
    \foreach \x in {0,...,\L}{
      \foreach \y in {0,...,\L}{
        \node at (2*\x,2*\y){$+$};
      }
    }
    \foreach \x in {0,...,\L}{
    \foreach \y in {-0,...,\L}{
		\draw (2*\x,2*\y-1) -- (2*\x+\c,2*\y-1);
		\node[minimum size=1mm] at (2*\x,2*\y+0.5){\tiny $+$};
		\node at (2*\x,2*\y-1){$=$};
		\node at (2*\x+\c,2*\y-1){};
      }
    }
    \foreach \x in {0,...,\L}{
      \foreach \y in {0,...,\L}{
		\draw (2*\x+1,2*\y) -- (2*\x+1,2*\y+\c);
		\node[minimum size=1mm] at (2*\x-.5,2*\y){\tiny $+$};
		\node at (2*\x+1,2*\y){$=$};
		\node at (2*\x+1,2*\y+\c){};
      }
    }
  \end{tikzpicture} 
	\caption{The dual NFG $\G'$.}
	\label{fig:dnfg}
\end{figure}

%% file: convergence.tex
\section{Convergence Behaviour of the Estimators}
\label{sec:converge}
Our analysis is primarily based on bounding the respective variances of the logarithm of the estimators for large $M$, as for any given number $M$ of samples, such variance is an indicator of the 
estimation accuracy. Our development is largely in line with that of \cite{Potamianos:1997}.


Given a statistical model NFG ${\cal G}$, it is possible to show 
\begin{eqnarray}
	\lim_{M\rightarrow \infty}M \var[\log(\zot_{\G}(M))] \hspace{-.6cm}&&=
	\frac{Z_{\G}^{2}}{|\X_{\G}|^{2}}\var\big[\frac{1}{f_{\G}(Y_1)}\big] \nonumber \\ && =
	\frac{Z_{\G}}{|\X_{\G}|^{2}}\sum_{x\in \X_{\G}}\frac{1}{f_{\G}(x)}-1.
\label{eq:zotvar}
\end{eqnarray}
\begin{proof}
	Let
	\[X_{M}:=\frac{\frac{1}{M}\sum\limits_{i=1}^{M}\frac{1}{f_{\G}(Y_{i})}}{|\X_{\G}|},\]
	then
	\[E[X_{M}]=\frac{1}{|\X_{\G}|}E[\frac{1}{f_{\G}(Y_{1})}]=\frac{1}{|\X_{\G}|}\sum_{x}\frac{p_{\G}(x)}{f_{\G}(x)}=\frac{1}{Z_{\G}},\]
	and
	\[\var[X_{M}]=\frac{1}{M|\X_{\G}|^{2}}\var[\frac{1}{f_{\G}(Y_{1})}].\]
	From (\ref{eq:zot}) we can rewrite $\log(\zot_{\G})$ as
	\[
	\log(\zot_{\G}) = g(X_{M}), 
	\]
	where $g(x):=\log(\frac{1}{x})=-\log(x)$, and so $g'(x)=\frac{-1}{x}$.
	Using Taylor expansion of $g$ at $E[X_{M}]$, 
	\[
	\log(\zot_{\G})\simeq g(E(X_{M})) + g'(E(X_{M}))(X_{M}-E(X_{M})),
	\]
	and so
	\begin{eqnarray*}
	\var[\log(\zot_{\G})]&& \simeq \big(g'(E(X_{M}))\big)^{2}\var[X_{M}] \\
	&& = \frac{1}{(E[X_{M}])^{2}} \var[X_{M}] \\
	&& = \frac{Z_{\G}^{2}}{M|\X_{\G}|^{2}} \var[\frac{1}{f_{\G}(Y_{1})}] \\
	\end{eqnarray*}
The approximation is only valid in the limit, as $g$ may be approximated as a linear function only when the
variance of $X_{M}$ is small.
This method of first order approximation is often referred to as the \emph{delta method}. For a more rigorous discussion 
on the delta method, see e.g. \cite[Theorem~5.5.24]{Casella:Inference}.
\end{proof}
Similarly, it can be shown that
\begin{eqnarray}
	\lim_{M\rightarrow \infty} M\var[\log\zunif_{\G}(M)] \hspace{-.6cm}&&=
	\frac{|\X_{\G}|^{2}}{Z^{2}_{\G}}\var[f_{\G}(Y_1)] \nonumber \\ && =
	\frac{|\X_{\G}|}{Z_{\G}^{2}}\sum_{x\in \X_{\G}}f_{\G}^{2}(x) - 1.
\label{eq:zunifvar}
\end{eqnarray}


From this, the following proposition can be proved.

\begin{proposition}
When sampling the NFG $\G$ of the Potts model, 
\begin{eqnarray*}
&&\Lot(\beta) \leq \lim_{M\rightarrow \infty}M \var[\log(\zot_{\G}(M))] \leq \Rot(\beta),  \\
&&\Lunif(\beta) \leq \lim_{M\rightarrow \infty}M \var[\log(\zunif_{\G}(M))] \leq \Runif(\beta), 
\end{eqnarray*}
where
\begin{eqnarray*}
	&& \Lot(\beta) :=\frac{e^{2N\beta}}{|\X_{\G}|^{2}}-1, \Rot(\beta):=e^{4N\beta}-1,\\
	&&\hspace{-.8cm} \Lunif(\beta) :=\frac{|\X_{\G}|}{(q+(|\X_{\G}|-q)e^{-8\beta})^{2}}-1, \Runif(\beta):=e^{8N\beta}-1. 
\end{eqnarray*}
\label{prop:p}
\end{proposition}
\begin{proof}
	We have
	\begin{eqnarray}
		\label{eq:fp}
	&&	e^{-2N\beta} \leq f_{\G}(x) \leq e^{2N\beta}, \\
		\label{eq:fip}
	&&	e^{-2N\beta} \leq \frac{1}{f_{\G}(x)} \leq e^{2N\beta},
	\end{eqnarray}
	and so,
	\begin{eqnarray}
		\label{eq:zp}
		e^{2N\beta}  \leq &Z_{\G}& \leq |\X_{\G}|e^{2N\beta}, \\
		\label{eq:zip}
		1 \leq & \sum_{x}\frac{1}{f_{\G}(x)} & \leq |\X_{\G}|e^{2N\beta}.
	\end{eqnarray}
(The lower bound in (\ref{eq:zip}) is trivially true, and made so to accommodate the case where the grid is of odd size
while keeping the derived bounds simple. If the grid is of even size, it can be replaced with $e^{2N\beta}$--- Color the
grid in black and white such that no similar colors are adjacent. The lower bound in (\ref{eq:zp}) is when all spins are
equal--- In fact there are $q$ such configurations and one may replace the lower bound with $qe^{2N\beta}$.)
Hence from (\ref{eq:zotvar}),
\begin{eqnarray*}
	\lim_{M\rightarrow \infty}M \var[\log(\zot_{\G}(M))] \hspace{-.6cm}&&= 
	\frac{Z_{\G}}{|\X_{\G}|^{2}}\sum_{x\in \X_{\G}}\frac{1}{f_{\G}(x)}-1 \\
	&& \stackrel{(\ref{eq:zip})}{\geq} \frac{Z_{\G}}{|\X_{\G}|^{2}} -1 \\
	&& \stackrel{(\ref{eq:zp})}{\geq} \frac{e^{2N\beta}}{|\X_{\G}|^{2}} -1,
\end{eqnarray*}
and
\begin{eqnarray*}
	\lim_{M\rightarrow \infty}M \var[\log(\zot_{\G}(M))] \hspace{-.6cm}
	&& \stackrel{(\ref{eq:zip})}{\leq} \frac{Z_{\G}}{|\X_{\G}|^{2}}|\X_{\G}|e^{2N\beta} -1 \\
	&& \stackrel{(\ref{eq:zp})}{\leq} \frac{|\X_{\G}|e^{2N\beta}}{|\X_{\G}|}e^{2N\beta} -1 \\
	&& = e^{4N\beta} -1. 
\end{eqnarray*}

From (\ref{eq:fp}) and (\ref{eq:fip}), we also have
\begin{eqnarray}
	\label{eq:zpu}
	|\X_{\G}|e^{-2N\beta}  \leq \hspace{-.6cm}&Z_{\G}& \hspace{-.6cm}\leq \big(q + (|\X_{\G}|-q)e^{-8\beta} \big) e^{2N\beta}, \\
	\label{eq:z2p}
	e^{4N\beta}  \leq &\sum_{x}f^{2}_{\G}(x)& \leq |\X_{\G}|e^{4N\beta}.
\end{eqnarray}
(In (\ref{eq:zpu}) we needed a tighter upper bound than in (\ref{eq:zp}).
Instead of trivially replacing the summand with its largest value, we kept its largest $q$ values and replaced the
summand's remaining values with its second largest value.)
From (\ref{eq:zunifvar})
\begin{eqnarray*}
	\lim_{M\rightarrow \infty}M \var[\log(\zunif_{\G}(M))] \hspace{-.6cm}&&= 
	\frac{|\X_{\G}|}{Z_{\G}^{2}}\sum_{x\in \X_{\G}}f_{\G}^{2}(x) - 1. \\
	&& \stackrel{(\ref{eq:z2p})}{\geq} \frac{|\X_{\G}|}{Z_{\G}^{2}} e^{4N\beta} - 1, \\
	&& \hspace{-.7cm}\stackrel{(\ref{eq:zpu})}{\geq} \frac{|\X_{\G}| e^{4N\beta}}{\big(q + (|\X_{\G}|-q)e^{-8\beta}
	\big)^{2} e^{4N\beta}}  - 1, 
\end{eqnarray*}
and
\begin{eqnarray*}
	\lim_{M\rightarrow \infty}M \var[\log(\zunif_{\G}(M))] \hspace{-.6cm}
	&& \stackrel{(\ref{eq:z2p})}{\leq} \frac{|\X_{\G}|}{Z_{\G}^{2}} |\X_{\G}|e^{4N\beta} - 1, \\
	&& \stackrel{(\ref{eq:zpu})}{\leq} \frac{|\X_{\G}|^{2}e^{4N\beta}}{|\X_{\G}|^{2}e^{-4N\beta}}  - 1.
\end{eqnarray*}
\end{proof}

We remark that
the bounds presented in the proposition above (and later in
Proposition~\ref{prop:d}) can be loose for some values of $\beta$. 
However, they suffice to explain the behaviour of the estimators on the
primal and dual NFGs.

When $\beta$ is small, say, in the order of $N^{-m}$ for $m>1$,  both upper bounds $\Rot$ and $\Runif$ in the
proposition approach zero with increasing $N$.  In this regime both estimators provide good estimates
of the partition function, without requiring asymptotically large $M$.

For large $\beta$, however, both estimators are inefficient. In particular, when $\beta>\log q$,  the lower bound $\Lot$
grows exponentially in $N$, which requires $M$ to be at least exponential in $N$ in order for the variance to be bounded
within a constant.  Similarly, when $\beta>\frac{\log q}{8}N$, the lower bound $\Lunif$ also grows exponentially in $N$,
making the uniform estimator inefficient. This is a rather exaggerated value of $\beta$, and we refer the reader to
\cite{Mackay:MC} for a better discussion on why the uniform estimator is inefficient for large $\beta$.
%

To get a better idea on relative performance between the OT and uniform estimators for large
$\beta$, note that
\[
\lim_{M\rightarrow \infty} M \var[\log(\zunif_{\G}(M))] \leq |\X_{\G}|-1,
\]
which follows immediately from the fact that $\sum_{x}f^{2}_{\G}(x) \leq Z_{\G}^{2}$. Comparing this upper bound with
the lower bound $\Lot$, there exists $\beta_{0}:=\frac{3}{2}\log(q)$ above which the uniform estimator is more efficient than
the OT estimator. 
This is in fact Theorem~2 of \cite{Potamianos:1997} for the model in this work.



On the dual side, we have the following bounds.

\begin{proposition}
For any integer $k$, let $A_{k,\beta}:=1+(k-1)e^{-2\beta}$, and let
$r(\beta):= \frac{A_{q,\beta}}{A_{0,\beta}}$.
When sampling the dual NFG $\G'$ for the Potts model (with $N$ being an even number), 
\begin{eqnarray*}
&&\Lot'(\beta) \leq \lim_{M\rightarrow \infty}M \var[\log(\zot_{\G'}(M))] \leq \Rot'(\beta),  \\
&&\Lunif'(\beta) \leq \lim_{M\rightarrow \infty}M \var[\log(\zunif_{\G'}(M))] \leq \Runif'(\beta), 
\end{eqnarray*}
where
\begin{eqnarray*}
	&& \Lot'(\beta) := \frac{r^{2N}(\beta)}{|\X_{\G'}|^{2}} -1, 
	\Rot'(\beta):= r^{2N}(\beta) - 1, \\
	&& \Lunif'(\beta) :=\frac{|\X_{\G'}|}{\big(q+(|\X_{\G'}|-q)A_{0,\beta} \big)^{2}} -1, \\
	&& \Runif'(\beta):= r^{4N}(\beta) - 1. 
\end{eqnarray*}
\label{prop:d}
\end{proposition}
\begin{proof}
	We have
\begin{eqnarray}
	\label{eq:fd}
	A_{0,\beta}^{2N}e^{2N\beta} \leq &f_{\G'}(x)& \leq A_{q,\beta}^{2N}e^{2N\beta}, \\
	\label{eq:fid}
	A_{q,\beta}^{-2N}e^{-2N\beta} \leq &\frac{1}{f_{\G'}(x)}& \leq A_{0,\beta}^{-2N}e^{-2N\beta}.
\end{eqnarray}
and so,
\begin{eqnarray}
	\label{eq:zd}
	A_{q,\beta}^{2N}e^{2N\beta} \leq &Z_{\G'}& \leq |\X_{\G'}| A_{q,\beta}^{2N}e^{2N\beta}, \\
	\label{eq:zid}
	A_{0,\beta}^{-2N}e^{-2N\beta} \leq &\sum_{x}\frac{1}{f_{\G'}(x)}& \leq |\X_{\G'}|A_{0,\beta}^{-2N}e^{-2N\beta}.
\end{eqnarray}
(The lower bound in (\ref{eq:zid}) is valid since the model is of even size.)
Hence from (\ref{eq:zotvar}),
\begin{eqnarray*}
	\lim_{M\rightarrow \infty}M \var[\log(\zot_{\G'}(M))] \hspace{-.6cm}&&= 
	\frac{Z_{\G'}}{|\X_{\G'}|^{2}}\sum_{x\in \X_{\G'}}\frac{1}{f_{\G'}(x)}-1 \\
	&& \stackrel{(\ref{eq:zid})}{\geq} \frac{Z_{\G'}}{|\X_{\G'}|^{2}} A_{0,\beta}^{-2N}e^{-2N\beta} -1 \\
	&& \stackrel{(\ref{eq:zd})}{\geq} \frac{A_{0,\beta}^{-2N}A_{q,\beta}^{2N}}{|\X_{\G'}|^{2}} -1,
\end{eqnarray*}
and
\begin{eqnarray*}
	\lim_{M\rightarrow \infty}M \var[\log(\zot_{\G'}(M))] \hspace{-.6cm}
	&& \stackrel{(\ref{eq:zid})}{\leq} \frac{Z_{\G'}}{|\X_{\G'}|^{2}}|\X_{\G'}|A_{0,\beta}^{-2N}e^{-2N\beta} -1 \\
	&& \stackrel{(\ref{eq:zd})}{\leq} {A_{q,\beta}^{2N}e^{2N\beta} A_{0,\beta}^{-2N}e^{-2N\beta}} -1 \\
	&& = A_{0,\beta}^{-2N} A_{q,\beta}^{2N}  -1.
\end{eqnarray*}

From (\ref{eq:fd}) and (\ref{eq:fid}), we also have
\begin{eqnarray}
	\label{eq:zdu}
	|\X_{\G'}|A_{0,\beta}^{2N}e^{2N\beta}  \leq  Z_{\G'}\leq 
	\big(q + (|\X_{\G'}|-q)A_{0,\beta} \big)A_{q,\beta}^{2N} e^{2N\beta}, \\
	\label{eq:z2d}
	A_{q,\beta}^{4N}e^{4N\beta}  \leq \sum_{x}f^{2}_{\G'}(x) \leq |\X_{\G'}|A_{q,\beta}^{4N}e^{4N\beta},
\end{eqnarray}
where the upper bound in (\ref{eq:zdu}) follows from
\begin{eqnarray*}
	Z_{\G'} \leq qA_{q,\beta}^{2N}e^{2N\beta} + (|\X_{\G'}|-q)A_{0,\beta}A_{q,\beta}^{2N-1}e^{(2N-1)\beta}.
\end{eqnarray*} 
From (\ref{eq:zunifvar})
\begin{eqnarray*}
	\lim_{M\rightarrow \infty}M \var[\log(\zunif_{\G'}(M))] \hspace{-.6cm}&&= 
	\frac{|\X_{\G'}|}{Z_{\G'}^{2}}\sum_{x\in \X_{\G'}}f_{\G'}^{2}(x) - 1 \\
	&& \stackrel{(\ref{eq:z2d})}{\geq} \frac{|\X_{\G'}|}{Z_{\G'}^{2}} A_{q,\beta}^{4N}e^{4N\beta} - 1 \\
	&& \hspace{-1.5cm}\stackrel{(\ref{eq:zdu})}{\geq}
	\frac{|\X_{\G'}|A_{q,\beta}^{4N}e^{4N\beta}}{\big(q + (|\X_{\G'}|-q)A_{0,\beta} \big)^{2}A_{q,\beta}^{4N} e^{4N\beta}}  - 1 \\
	&& \hspace{-1.5cm} = \frac{|\X_{\G'}|}{\big(q + (|\X_{\G'}|-q)A_{0,\beta} \big)^{2}}  - 1. \\
\end{eqnarray*}
and
\begin{eqnarray*}
	\lim_{M\rightarrow \infty}M \var[\log(\zunif_{\G'}(M))] \hspace{-.6cm}
	&& \stackrel{(\ref{eq:z2d})}{\leq} \frac{|\X_{\G'}|}{Z_{\G'}^{2}} |\X_{\G'}|A_{q,\beta}^{4N}e^{4N\beta} - 1, \\
	&& \stackrel{(\ref{eq:zdu})}{\leq} \frac{|\X_{\G'}|^{2} A_{q,\beta}^{4N}e^{4N\beta}}{|\X_{\G'}|^{2} A_{0,\beta}^{4N}e^{4N\beta}}  - 1, \\
	&& = A_{0,\beta}^{-4N}A_{q,\beta}^{4N} - 1.
\end{eqnarray*}
\end{proof}

When $\beta$ is large, namely in the order of $\log(N)$, both upper bounds $\Rot'$ and $\Runif'$ in the proposition
approach zero with increasing $N$. In this regime both estimators provide good estimate of the partition function,
without requiring asymptotically large $M$.

For small $\beta$, however, both estimators are inefficient. In particular, for $\beta <
\frac{1}{2}\log\big(\frac{2q-1}{q-1}\big)$, the lower bound $\Lot'$ grows exponentially in $N$, which requires
$M$ to be at least exponential in $N$ in order for the variance to be bounded within a constant.
Similarly, since $A_{0,\beta}$ approaches zero when $\beta$ approaches zeros, $\Lunif'$ becomes exponential in
$N$.


Similar to the remark following Proposition~\ref{prop:p}, 
comparing $|\X_{\G'}|+1$ with the lower bound $\Lot'$, it follows that there exists
$\beta'_{0}:=\frac{1}{2}\log(1+\frac{q}{q^{2}-1})$ \emph{below} which the uniform estimator is more efficient than
the OT estimator. 

At this end, we have shown that on the dual NFG, the two estimators behave in an opposite trend (in $\beta$) to that on
the primal NFG. It appears that such a phenomenon may  fundamentally be related to a ``duality'' between ``nearly
uniform'' and ``nearly concentrated'' distribution. More precisely, when both an NFG and its dual involve only
non-negative local functions, they both can be associated with a Boltzmann distribution. If one of the distributions is
``nearly uniform'', the other one is necessarily ``nearly concentrated'', namely, assigning most of the probability mass
to only a few configurations. It is well-known in physics literature that the ``near uniformity'' and ``near
concentratedness'' correspond respectively to high-temperature and low-temperature systems respectively.  It appears
that these sampling based estimators usually work well for high-temperature systems and work poorly for low-temperature
systems. Taking an NFG to its dual, essentially reverts the ``temperature''.


%% file: simulations.tex
\section{Experiments}
\label{sec:experiments}

In this section we provide experimental results for the Potts model with $q=4$ and grid-size $N = 10\times 10$.
We use the Gibbs sampling algorithm
\cite{Newman:MC}
on the primal and dual NFG to obtain samples from $p_{\G}$ and $p_{\G'}$, respectively.
%
We estimate the log partition function per site (i.e. $\frac{1}{N}\log(Z)$), where
depending on whether the primal or the dual NFG is used, the estimate of the partition function, which
depends on the number of samples $M$, is defined as 
$\hat{Z}(M):=\zot_{\G}(M)$ and $\hat{Z}(M):=q^{N}\zot_{\G'}(M)$, respectively.
(Similar definitions are used for the uniform estimator.)
For any number of samples $M$, we repeat the experiment $30$ times and record the value of
$\frac{1}{N}\log(\hat{Z}_{i}(M)), i=1,\cdots,30$, where for each trial $i$, the initial configuration is chosen independently
and according to the uniform distribution.
The ``quality'' of the estimation at any $M$ is decided based on the standard deviation of the trials from their mean
(with respect to the uniform distribution on the set $\{1, \cdots, 30\}$).

Figs~\ref{fig:b1p2p} and \ref{fig:b1p2d} show the estimated log partition function per site, i.e.,  $\log(\hat{Z}(M))/N$, for
the low temperature $\beta=1.2$
Fig.~\ref{fig:b1p2p} shows the estimation based on the primal
NFG using both the uniform estimator (left) and the OT estimator (right). 
Using up to $10^{6}$ samples, both estimators fail to converge, and so do not provide a good estimation.
This can also be seen in the dashed lines in Fig.~\ref{fig:b1p2std} showing the standard deviation of the uniform estimator (left) and
the OT estimator (right), where the standard deviation in both cases remains high.
In contrast, Fig.~\ref{fig:b1p2d} shows fast convergence of the estimators on the dual NFG. The
standard deviation of the estimations obtained from the dual NFG is shown using the solid lines in Fig.~\ref{fig:b1p2std}.
Fig.~\ref{fig:bp18} shows the standard deviation of the estimations for the high temperature of $\beta=0.18$. 
In this case estimations based on the primal NFG have a lower standard deviation compared to the dual NFG, and so
provide a better estimation.
In Fig.~\ref{fig:bv}~(a), showing the standard deviation versus $\beta$ using uniform sampling, one observes the behaviour 
of the estimator versus $\beta$ as discussed in Section~\ref{sec:converge}.

\begin{figure*}[ht]
	\centering
	\def\s{.48}
	{\includegraphics[scale=\s]{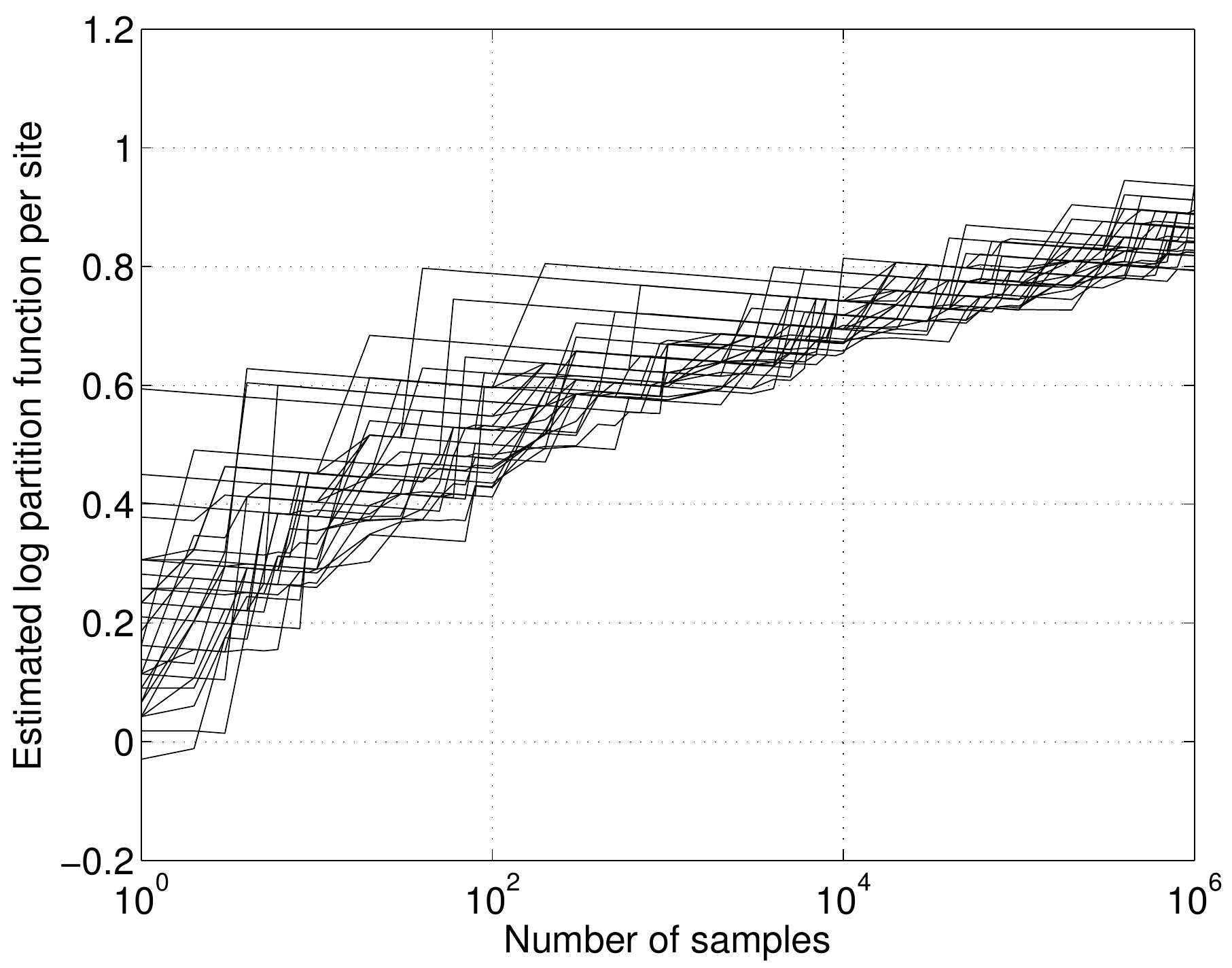}		
      \includegraphics[scale=\s]{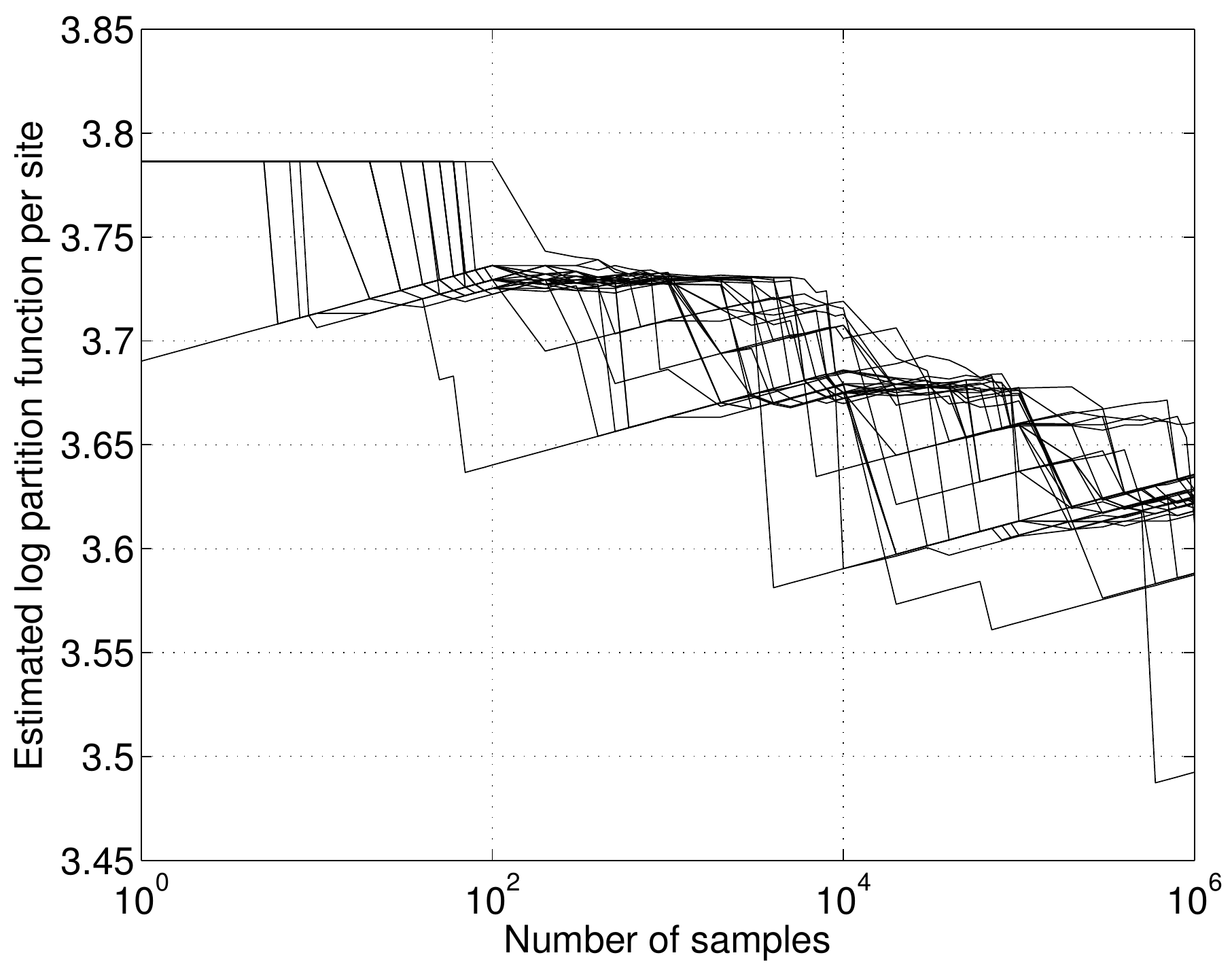}		}
	\caption{Potts model at low temperature $\beta=1.2$ using the primal NFG. The two figures show the estimated log partition function per site
	versus the number of samples using the uniform estimator (left) and the OT estimator (right), where each line represents a trial. }
	\label{fig:b1p2p}
 \end{figure*} \begin{figure*}[ht] \def\s{.48} \centering
	{\includegraphics[scale=\s]{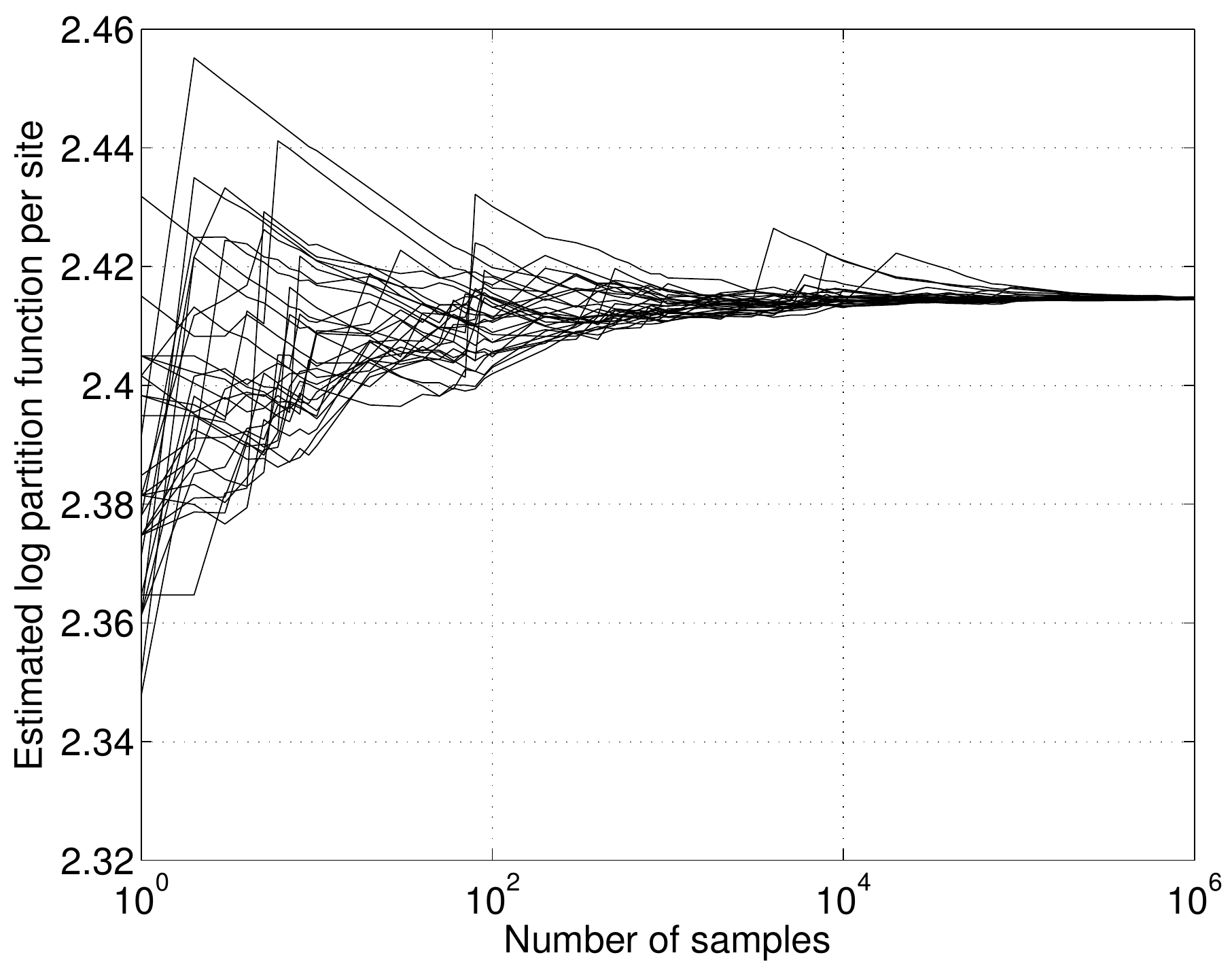}		
      \includegraphics[scale=\s]{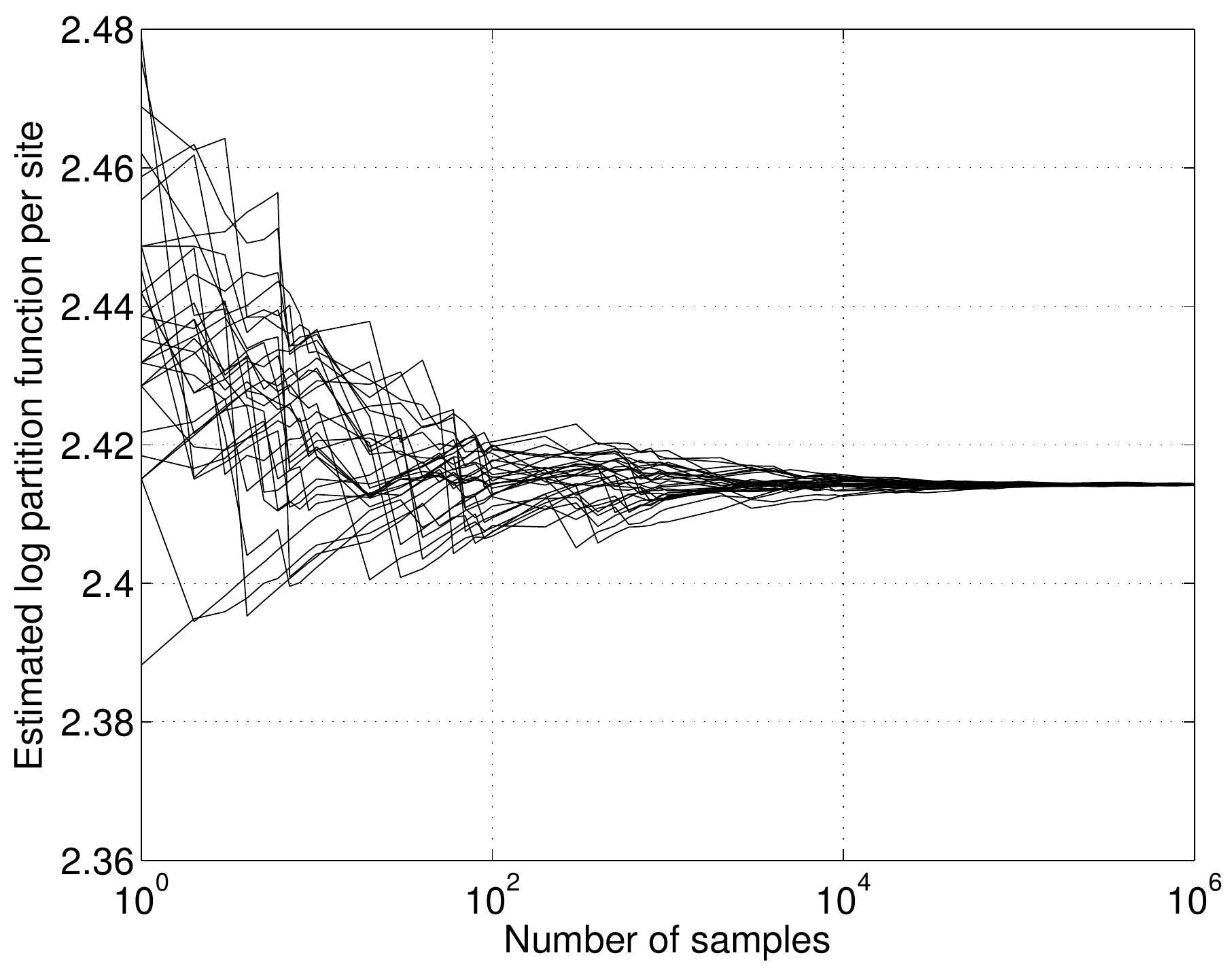}		}
	\caption{Potts model at low temperature $\beta=1.2$ using the dual NFG. The two figures show the estimated log partition function per site
	versus the number of samples using the uniform estimator (left) and the OT estimator (right), where each line represents a trial. }
	\label{fig:b1p2d}
 \end{figure*} \begin{figure*}[ht] \def\s{.48} \centering
	{\includegraphics[scale=\s]{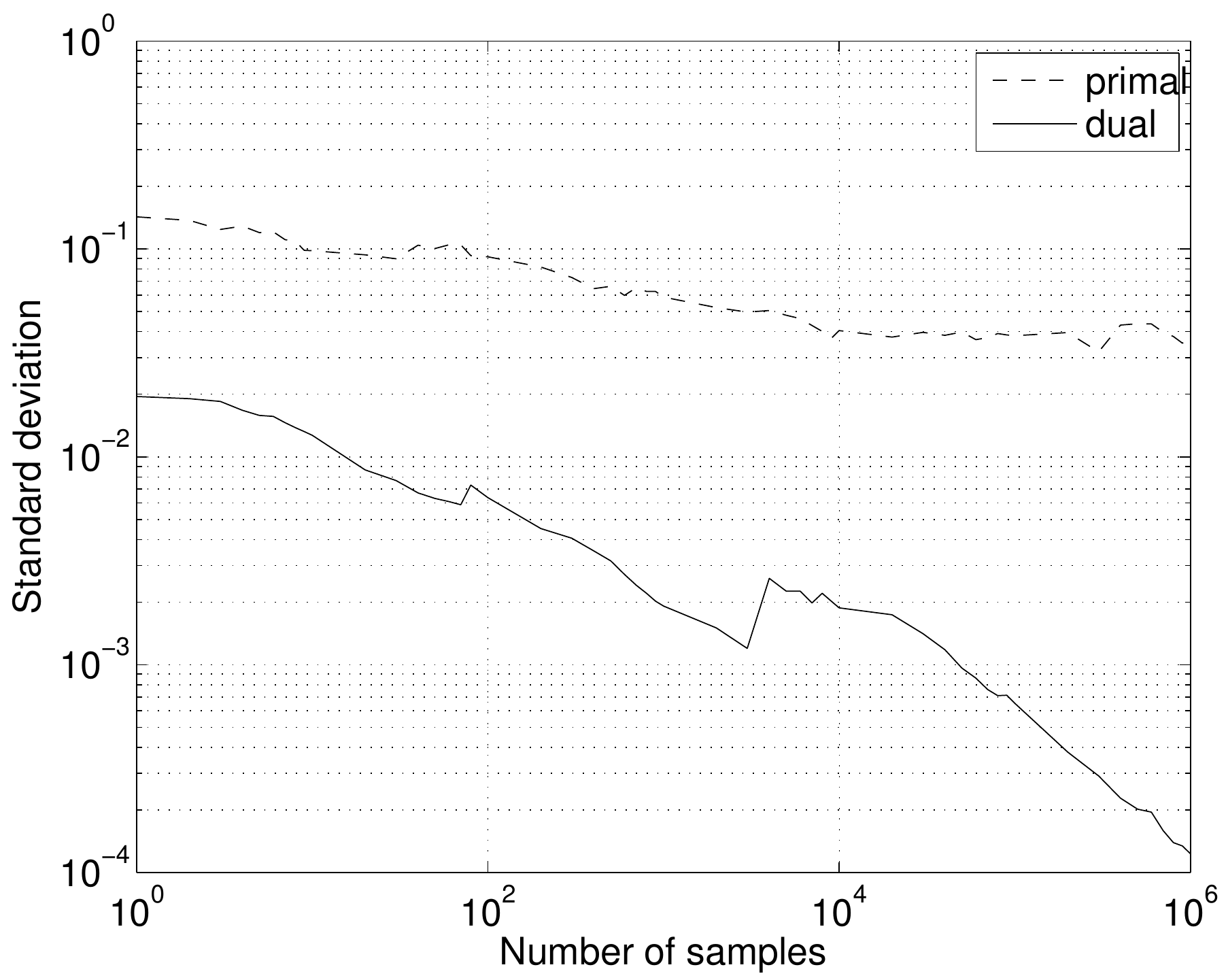}		
      \includegraphics[scale=\s]{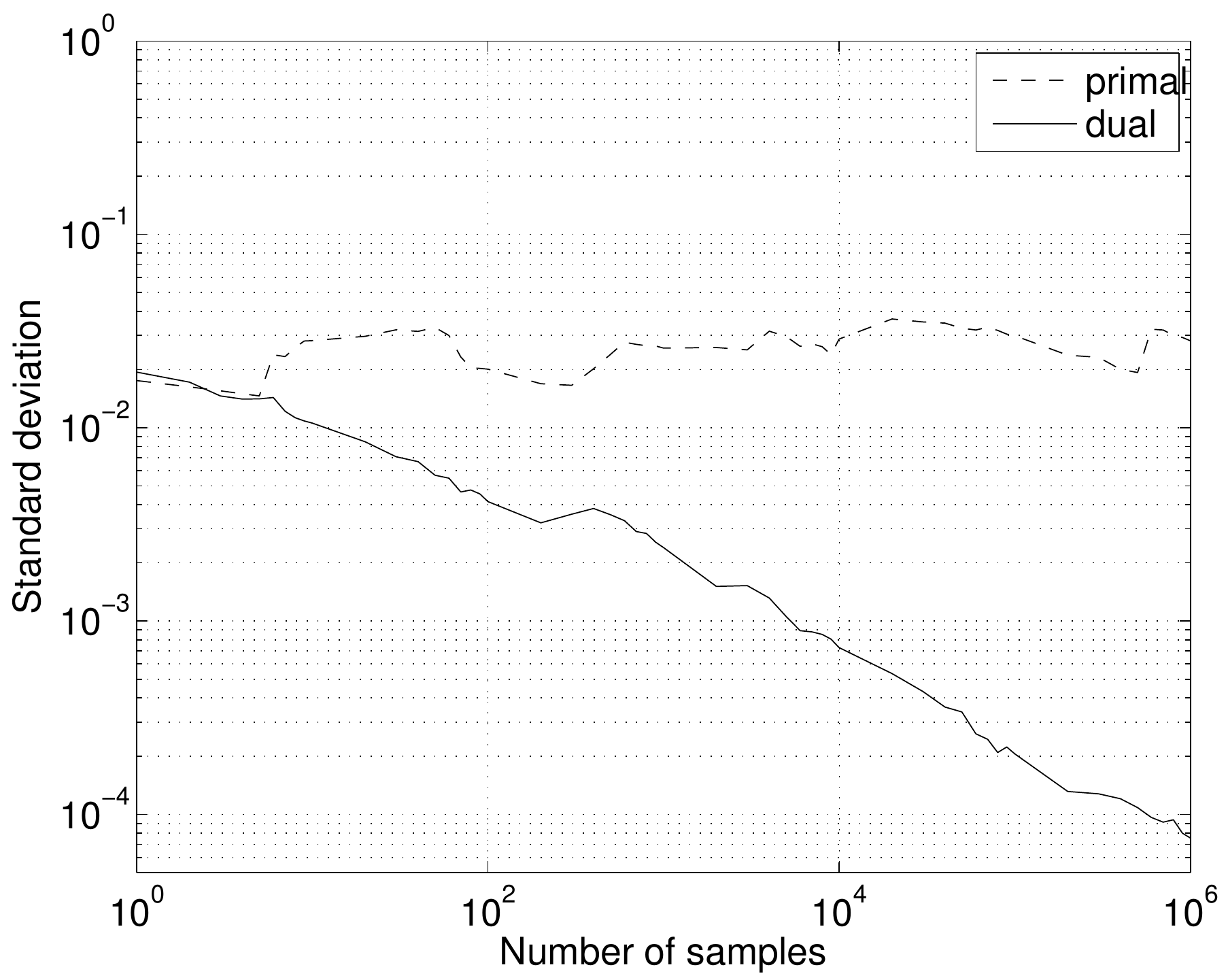}		}
	\caption{
	Potts model at low temperature $\beta=1.2$, where the standard deviation of the estimated log partition function per site
	is shown for the uniform (left) and OT (right) estimators based on the primal (dashed line) and dual (solid line)
	NFGs.}
	\label{fig:b1p2std}
\end{figure*}

\begin{figure*}[ht]
	\centering
	\def\s{.48}
	{\includegraphics[scale=\s]{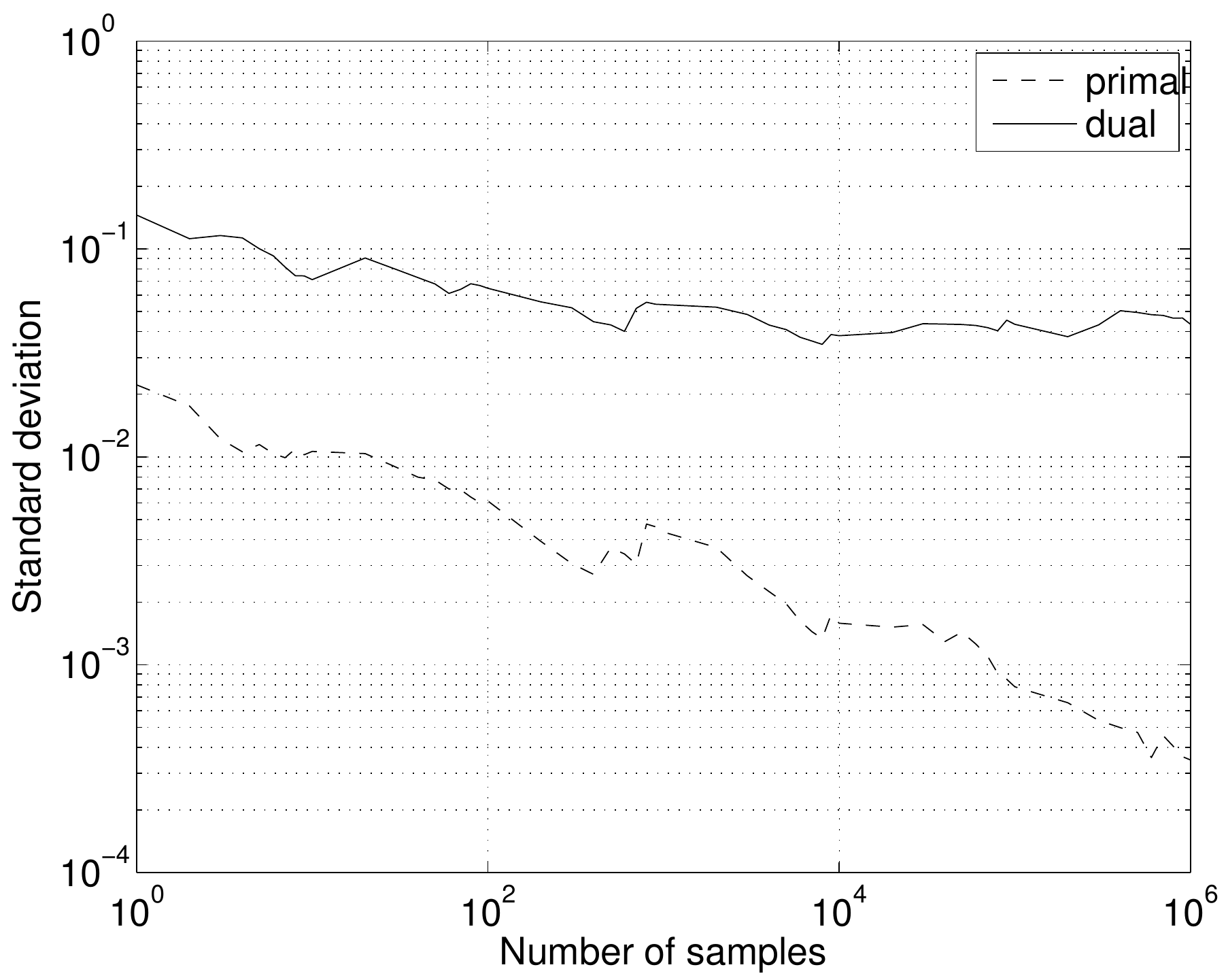}		
      \includegraphics[scale=\s]{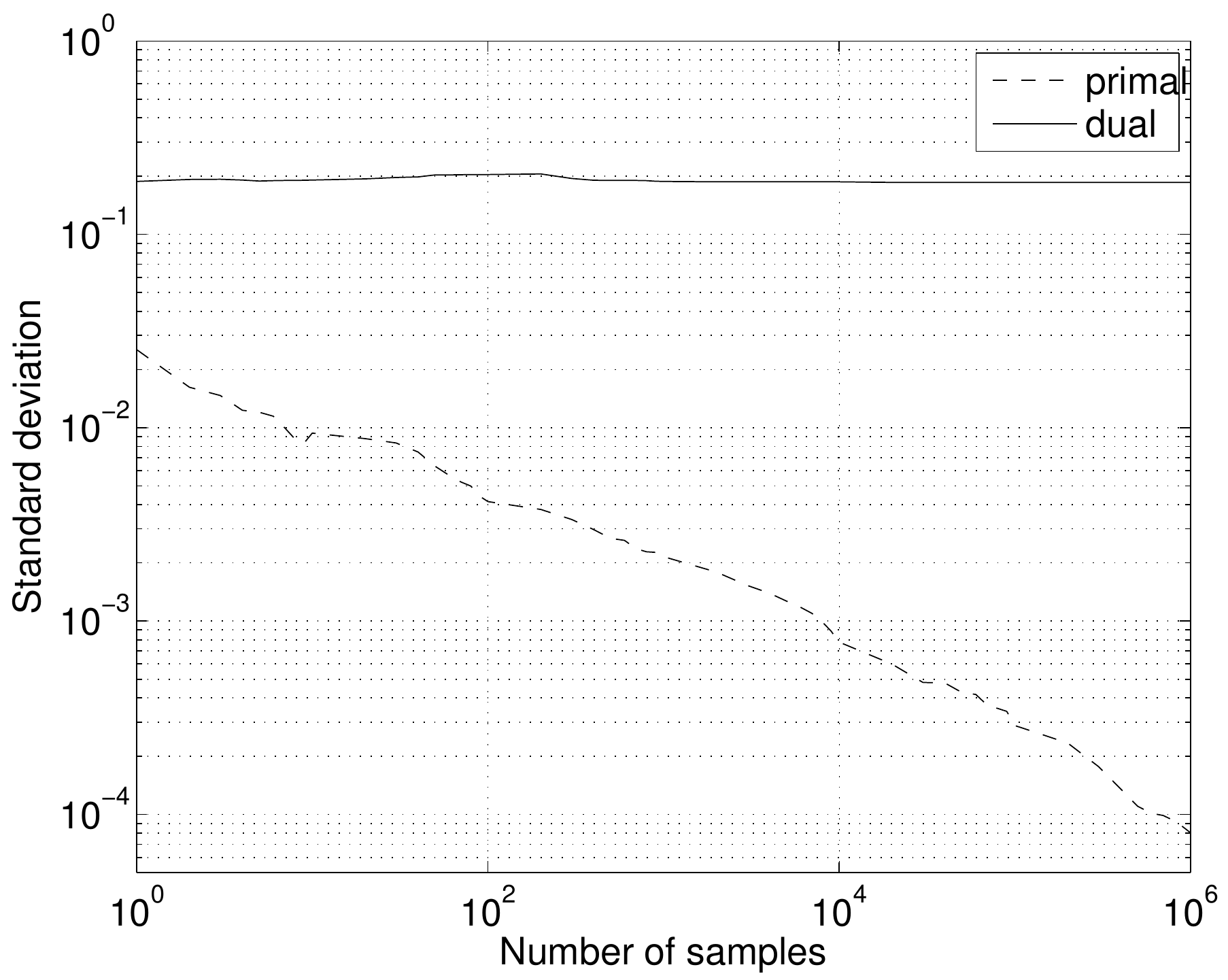}
	 }
	\caption{
	Potts model at high temperature $\beta=0.18$, where the standard deviation of the estimated log partition function per site
	is shown for the uniform (left) and OT (right) estimators based on the primal (dashed line) and dual (soled line)
	NFGs.}
	\label{fig:bp18}
\end{figure*}

\begin{figure*}[ht]
	\centering
	\def\s{.48}
	\subfigure[Potts model.]{
	\includegraphics[scale=\s]{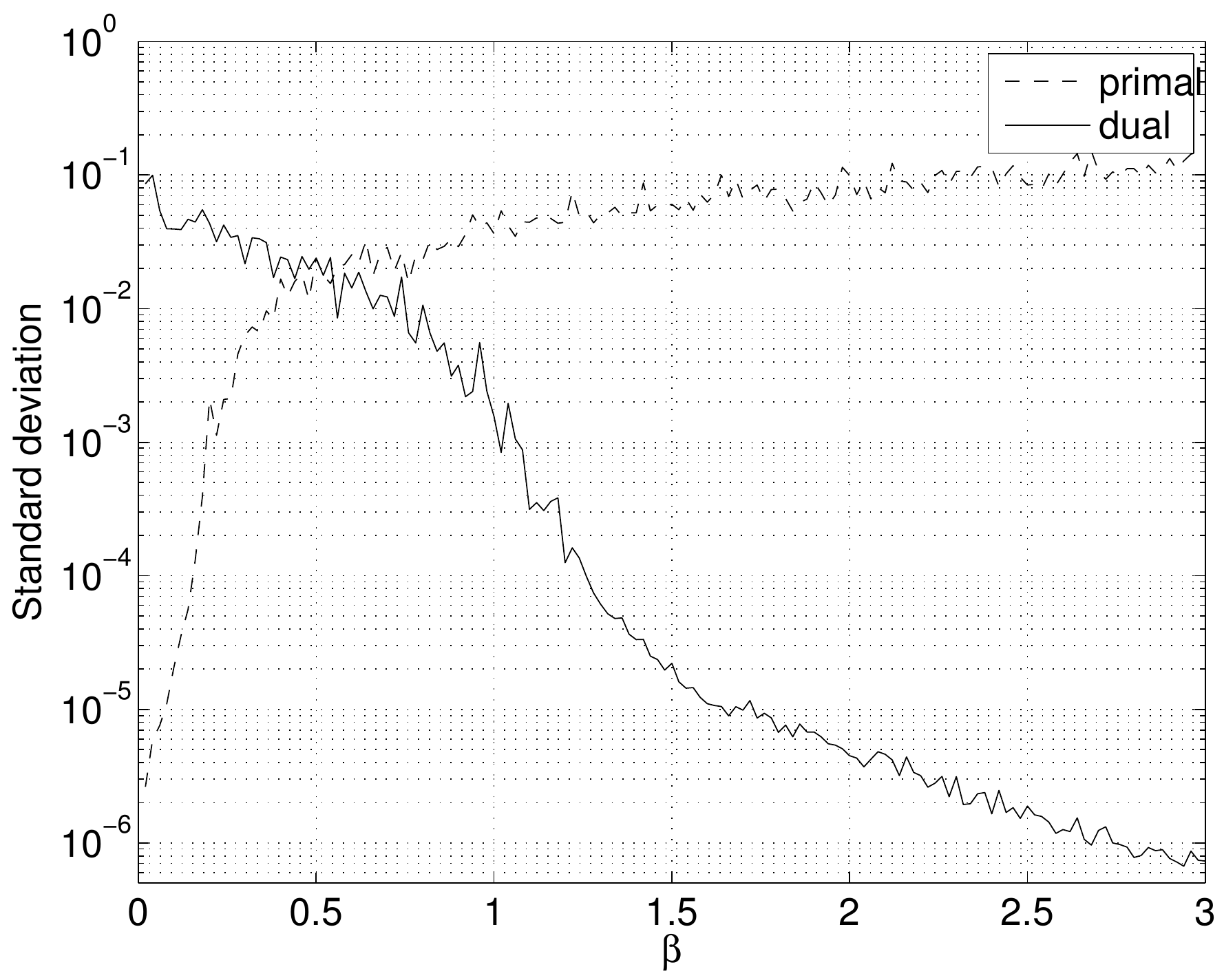}		}
	\subfigure[Clock model.]{
	\includegraphics[scale=\s]{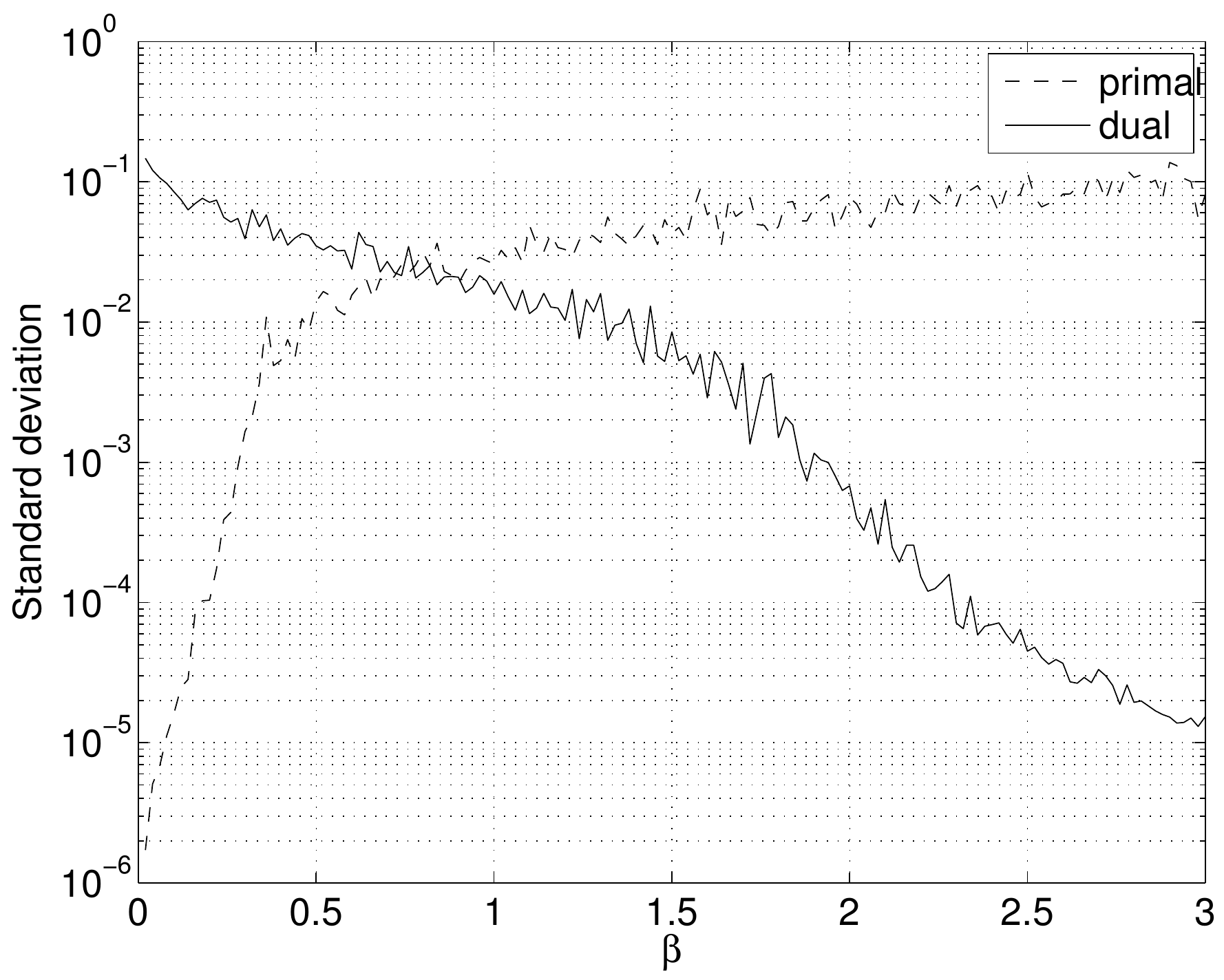}		}
	\caption{Standard deviation of the estimated free energy per site versus $\beta$ using uniform sampling with $M = 10^{6}$.}
	\label{fig:bv}
\end{figure*}


%% file: conclude.tex
\section{Concluding Remarks: Beyond the Potts model}
\label{sec:conclude}

This paper shows analytically and experimentally that stochastic estimators of partition functions
exhibit opposite trends on NFG representation of a model and its dual. As remarked in 
Section~\ref{sec:converge}, this phenomenon is fundamentally related to a duality between ``nearly concentrated'' and
``nearly uniform'' distribution. This understanding allows the results presented above to extend beyond the Potts models.
In particular, one may 
consider two-dimensional nearest neighbor models whose bivariate local function is of the form
$h(x,x'):=\kappa(x-x')$ for other functions $\kappa$. When both $\kappa$ and $\widehat{\kappa}$ are a non-negative real function, the duality between uniformity and concentratedness is expected to hold and such a phenomenon is expected to occur.
%
%
As an example, consider the ``clock model,'' which is defined in the same way as the Potts model under the choice 
\begin{eqnarray}
	\kappa_{\rm clock}(x) = e^{\beta\cos(2\pi x/q)},
	\label{eq:lee}
\end{eqnarray}
for all $x\in \Z_{q}$.
(Hence, it is within the scope of models of Fig.~\ref{fig:pnfg2}.)
From Lemma~\ref{lemma:clock} below,  $\widehat{\kappa}_{\rm clock}$ is a positive function, and so it is possible to
take the dual NFG route toward estimating its partition function.
\begin{lemma}
	$\widehat{\kappa}_{\rm clock}$ is a postitive function.
	\label{lemma:clock}
\end{lemma}
\begin{proof}
	For any $x,y \in \Z_{q}$, let $\chi_{y}(x):=e^{2\pi\sqrt{-1}xy/q}$.
Using Taylor expansion, we have
\[
\kappa_{\rm c}(x) = \sum_{n=0}^{\infty} \frac{\beta^{n}g_{n}(x)}{n!},
\]
where 
\begin{eqnarray*}
	g_{n}(x)\hspace{-.6cm}&&
	:=\cos^{n}(2\pi x/q) 
	=\frac{1}{2^{n}}\big(\chi_{1}(x)+\chi_{1}(-x)\big)^{n} \\
	&&= \frac{1}{2^{n}}\sum_{l=0}^{n} \binom{n}{l} \chi_{1}^{n-l}(x) \chi_{1}^{l}(-x)  \\
	&&= \frac{1}{2^{n}}\sum_{l=0}^{n} \binom{n}{l} \chi_{\overline{n}-2\overline{l}}(x),
\end{eqnarray*}
where for any $m\in \Z$, $\overline{m}\in \Z_{q}$ is defined as $m$ modulo $q$.
Hence,
\[
\widehat{g_{n}}(\chi_{k}) = \frac{q}{2^{n}}\sum_{l=0}^{n}\binom{n}{l}[\overline{n}+k-2\overline{l} = 0]
\]
is a non-negative function that is upper bounded by $q$, where for any $m$, $[m=0]$ is the indicator function evaluating
to one iff $m=0$.
Therefore,
\[
\widehat{\kappa}_{\rm c}(\chi_{k}) = \sum_{n=0}^{\infty} \frac{\beta^{n}\widehat{g_{n}}(\chi_{k})}{n!}
\]
is a positive function for $\beta > 0$. (This follows since for any $k\in \Z_{q}$, $\widehat{g_{n}}(\chi_{k})$ cannot be zero for all $n$.
In particular, $\widehat{g_{q-k}}(\chi_{k}) > 0$.) 
Finally, the series in the RHS is convergent since
\[
\sum_{n=0}^{\infty} \frac{\beta^{n}\widehat{g_{n}}(\chi_{k})}{n!}
\leq q\sum_{n=0}^{\infty} \frac{\beta^{n}}{n!} = q e^{\beta}.
\]
\end{proof}


Consider for instance the clock model with $q=4$. It is not hard to see that $p_{\G}$ in this case is a ``concentrated''
distribution for low temperatures and an ``almost uniform'' distribution for high temperatures. From this and the fact
that
\begin{eqnarray}
	\widehat{\kappa}_{\rm clock}(x) = \left\{ \begin{array}{ll}
		e^{\beta}+e^{-\beta}+2, & x = 0 \\
		e^{\beta}-e^{-\beta}, & x \in \{1,3\} \\
		e^{\beta}+e^{-\beta}-2, & x = 2, 
	\end{array} \right.
	\label{eq:dlee}
\end{eqnarray}
one may obtain similar results to Propositions~\ref{prop:p} and \ref{prop:d}. Simulation results for this
model are shown in Fig.~\ref{fig:bv}~(b).